\documentclass{amsart}
\usepackage{amsfonts,xcolor}
\usepackage{amsthm}
\usepackage{latexsym}
\usepackage{amsmath}
\usepackage{amssymb}
\usepackage{arydshln}
\newtheorem{theorem}{Theorem}
\theoremstyle{plain}

\newtheorem{corollary}{Corollary}

\newtheorem{example}{Example}

\newtheorem{lemma}{Lemma}

\newtheorem{remark}{Remark}

\numberwithin{equation}{section}

\begin{document}
\title{On the regularity of $\{\lfloor\log_b(\alpha n+\beta)\rfloor\}_{n\geq0}$}
\author{Jiemeng Zhang, Yingjun Guo and Zhixiong Wen}

\begin{abstract}
Let $\alpha,\beta$ be real numbers and $b\geq2$ be an integer. Allouche and Shallit showed that the sequence $\{\lfloor\alpha n+\beta\rfloor\}_{n\geq0}$ is $b$-regular if and only if $\alpha$ is rational. In this paper, using a base-independent regular language, we prove a similar result that the sequence $\{\lfloor\log_b(\alpha n+\beta)\rfloor\}_{n\geq0}$ is $b$-regular if and only if $\alpha$ is rational. In particular, when $\alpha=\sqrt{2},\beta=0$ and $b=2$, we answer the question of Allouche and Shallit  that the sequence $\{\lfloor\frac{1}{2}+\log_2n\rfloor\}_{n\geq0}$ is not $2$-regular, which has been proved by Bell, Moshe and Rowland respectively.
\end{abstract}

\maketitle


\section{Introduction}
Automatic sequence is studied by many authors both in formula language theory and number theory. It has many descriptions \cite{AS03,CKM,Cob,GW14}.  One of them is that its $k$-kernel is finite. Precisely, we call a sequence $\{u_n\}_{n\geq0}$ is \emph{$k$-automatic} if the set of subsequences  $\{\{u_{k^in+j}\}_{n\geq0}:i\geq0,0\leq j\leq k^i-1\}$ is finite.

But automatic sequences are defined over a finite alphabet. To overcome this limit,  Allouche and Shallit \cite{Jp,jp} generalized the concept of automatic sequence to regular sequence where the sequences may take infinitely many values. We call a sequence $\{u_n\}_{n\geq0}$ is \emph{$k$-regular} if the $R$-module generated by the set $\{\{u_{k^in+j}\}_{n\geq0}:i\geq0,0\leq j\leq k^i-1\}$ is finitely generated, where $R$ is a commutative Noetherian ring.

Many properties of regular sequence have been studied \cite{Jp,jp,Minimal,Note,Sum,AFO}. In \cite{Jp}, Allouche and Shallit proved that a sequence is $k$-regular and takes finitely many values if and only if it is $k$-automatic. Moreover, if $\{u_n\}_{n\geq0}$ is a  $k$-regular sequence,  they showed that there exists a constant $c>0$ such that $u_n=\mathcal{O}(n^c)$. If $\{u_n\}_{n\geq0}$ is unbounded, then, Bell, Coons and Hare \cite{Minimal} showed that there exists a constant $c>0$ such that $|u_n|>c\log n$ infinitely often.  Let $\{u_n\}_{n\geq0}$ be an integer sequence and there exist $c,c_1,c_2>0$ such that $c_1\log n<|u_n|\leq c_2 n^c$. It is interesting to determine the regularity of $\{u_n\}_{n\geq0}$.

In 2003, Allouche and Shallit \cite{jp} listed many examples. One typical example is that the sequence $\{\lfloor\alpha n+\beta\rfloor\}_{n\geq0}$ is $b$-regular if and only if $\alpha$ is rational, where~$\alpha,\beta\in\mathbb{R}$ and $b\geq 2$ is an integer. And at the end of this paper, they asked a question: is the sequence $\{\lfloor\frac{1}{2}+\log_2 n\rfloor\}_{n\geq0}$ 2-regular? Until 2008, Bell \cite{Bell}, Moshe \cite{YM} and Rowland \cite{Row} answered this question in different way. Moreover, Bell \cite{Bell} and Moshe \cite{YM} showed the sequence~$\{\lfloor\log_b\alpha n\rfloor\}_{n\geq0}$ is $b$-regular if and only if~$\alpha\in\mathbb{Q}$.

In this paper, we give a general result and state it as follows.

\begin{theorem}\label{main}
Let~$\alpha,\beta\in\mathbb{R}$ and $b\geq 2$ be an integer.   The sequence $\{\lfloor\log_b(\alpha n+\beta)\rfloor\}_{n\geq0}$ is $b$-regular if and only if $\alpha$ is rational.
\end{theorem}

Clearly, for any $\alpha,\beta\in\mathbb{R}$ and $b\geq 2$, there exist $c_1,c_2>0$ s.t., $c_1\log n<\lfloor\log_b \alpha n+\beta\rfloor\leq c_2 n$. Taking $\beta=0$, we obtain Bell and Moshe's reuslt. Moreover, taking $\alpha=\sqrt{2},b=2$, we answer the question of Allouche and Shallit. 
Different from previous proof, we construct a language by base-changing and prove that its regularity is independent of the base.

The organization of this paper is as follows. In setion 2, we recall some definitions and notation. In section 3, using base-changing, we construct a language and prove that its regularity is independent of the base. In section 4, we prove Theorem \ref{main}.

\section{Preliminary}
In this section, we recall some definitions and notation. More details, please see \cite{AS03,VM,F}.

For any integer $b\geq2$, set $\Sigma_b :=\{0,1,\cdots,b-1\}$. Let~$\Sigma_b^k$ be
the set of words of length~$k$ over~$\Sigma_b$ and~$\Sigma_b^{\ast}=\bigcup_{k%
\geq0}\Sigma_b^k$. If $w\in\Sigma_b^{\ast}$, then its length is denoted  by $|w|$. If $|w|=0$, then $w$ is called to be the empty word, denoted by $\epsilon$. For any finite words $u=u_0u_1\cdots u_m,v=v_0v_1\cdots v_n$, their \emph{concatenation}, denoted by $uv$, is $u_0u_1\cdots u_mv_0v_1\cdots v_n.$  For any word $u,v\in\Sigma_b^{\ast}$ with $|v|\geq1$, we always call that the word $v$ is a \emph{suffix}  of word $uv$. Clearly that the set $\Sigma_b^{\ast}$ together with concatenation forms a monoid, where the empty $\epsilon$ playing the part of the identity element.

Let $n\geq0$ be an integer, then there is an unique representation of the form $n=\sum_{i=0}^{m}u_i b^i$ with $u_m\neq0$ and $u_i\in\Sigma_b$. We call $u_mu_{m-1}\cdots u_0$ is its~\emph{$b$-ary expansion}, denoted by $(n)_b$. Conversely, for any finite word $W=w_0w_1\cdots w_n$, we define an integer $[W]_b=[w_0w_1\cdots w_n]_b:=\sum_{i=0}^{n}w_i b^{n-i}.$ Note that for any integer $n\geq 0$, $[(n)_b]_b=n$. For any real number $r$, the symbols $\lfloor r\rfloor, \{r\}$ denote its integral part and its fraction part respectively.

A subset of $\Sigma_b^{\ast}$ is called to be a language. Since a language is a set of words, we can apply all the usual set operations
like union, intersection or set difference: $\cup,\cap$ or $\setminus$. If $L_1,L_2$ are languages, then their \emph{concatenation} is $L_1L_2:=\{uv:u\in L_1,v\in L_2\}.$ Set $L^0=\{\epsilon\}$, then the \emph{Kleene star} of the language $L$ is defined as $L^*:=\bigcup_{n\geq0}L^n.$
If a language $L$ can be obtained by applying to some finite languages: a finite number of operations of union, concatenation and Kleene star, then this language is said to be a \emph{regular language}.  Note that any finite set is a regular language and regular language is closed under intersections, unions and complements\footnote{Let~$L\subseteq\Sigma^{\ast}$~be a language, its complement is defined as~$\overline{L}
=\Sigma^{\ast}\backslash L$.}. Moreover, if $\sigma$ is a homomorphism and $L$ is regular, then
$$\sigma(L):=\{\sigma(w):w\in L\},\quad \sigma^{-1}(L):=\{w:\sigma(w)\in L\}$$ are both regular. 

A set~$S\subseteq\mathbb{N}$~of integers is~\emph{$b$-recognisable} if the language $\{(n)_b:~n\in S\}$ is regular.
Observe that a set~$S$~is~$b$-recognisable if and only if its characteristic word
$$\mathcal{X}_S(n)=\left\{
    \begin{array}{ll}
      1, & \hbox{if~$n\in S$;} \\
      0, & \hbox{otherwise.}
    \end{array}
  \right.
$$ is
~$b$-automatic.

The following lemma, known as the pumping lemma, gives an important description of regular language.
\begin{lemma}[Pumping Lemma]
Let~$L\subseteq\Sigma_b^{\ast}$~be a regular language. Then
there exists a constant~$n\geq1$ such that for all strings~$Z\in L$~with~$|Z|\geq n$, there exists
a decomposition~$Z=UVW$, where~$U,V,W\in\Sigma_b^{\ast}$~and~$|UV|\leq n$~and~$|V|\geq1$,
such that~$UV^iW\in L$~for all~$i\geq0$.
\end{lemma}

\section{Base-independent regular language}
Let $B\geq 2$ be an integer and $\mathbf{u}=\{u_i\}_{i\geq0}$ be an integer sequence with $u_i\in\{0,1,\cdots,B-1\}$ for all $i\geq0$. For any integer $b\geq2$, define a language $$L_b(\mathbf{u}):=\left\{\left([u_0u_1u_2\cdots u_n]_b\right)_b\in \Sigma_b^*:n\geq0\right\}.$$

In this section, we study the language $L_b(\mathbf{u})$ which is constructed from a sequence by base-changing. Note that if $B\leq b$, then $([u_0u_1u_2\cdots u_n]_b)_b=u_0u_1u_2\cdots u_n$ for any integer $n\geq0$. Using pumping lemma, Moshe proved in \cite{YM} that the language $\{u_0u_1u_2\cdots u_n:n\geq0\}$ is regular if and only if $\mathbf{u}$ is ultimately periodic.
In this paper, we prove that the regularity of $L_b(\mathbf{u})$ is independent of the base $b$.
\begin{theorem}\label{regular language}
For any integer $b\geq2$, the language $L_b(\mathbf{u})$ is regular if and only if the sequence $\mathbf{u}$ is ultimately periodic.
\end{theorem}
\begin{corollary}
For any finite word $W\in \Sigma_b^*$, $L_b(\mathbf{u})$ is regular if and only if $L_b(W\mathbf{u})$ is regular.
\end{corollary}

\begin{proof}[Proof of Theorem \ref{regular language} ]
If $\mathbf{u}$ is a ultimately periodic sequence, then we obtain easily that $L_b(\mathbf{u})$ is regular. Conversely, fix $b$ and assume $L_b(\mathbf{u})$ is regular, we will prove $\mathbf{u}$ is ultimately periodic.

Let $v_n=[u_0u_1u_2\cdots u_n]_b$ for all $n\geq0$. We claim that there exists an integer $N$ such that
\begin{equation}\label{1}
  |(v_{n+1})_b|=|(v_{n})_b|+1,~n>N.
\end{equation}

Set $C:=\lfloor\frac{B-1}{b-1}\rfloor+1$, choose a large integer $K$ satisfying $b^{K+1}-b+B\leq b^{K+2}-C$. If there exist infinitely many $N$ such that $|(v_{N+1})_b|-|(v_{N})_b|\geq2,$ then there exists an integer~$N$~such that~$b^{K-1} \leq v_N\leq b^K-1$~but~$v_{N+1}\geq b^{K+1}$. Hence we have~$b^{K+1}\leq v_{N+1}=bv_N+u_{N+1}\leq b^{K+1}-b+u_{N+1}\leq b^{K+1}-b+B\leq b^{K+2}-C.$ By the choice of $C$, it is easy to check that ~$b^{K+i} \leq v_{N+i}\leq b^{K+i+1}-C$ for all $i\geq1$. Thus,  the claim $(\ref{1})$~holds.

Claim $(\ref{1})$ tells us that for any large integer $n$, there is exactly one word with length $n$ in $L$. Since $L_b(\mathbf{u})=\{(v_n)_b:n\geq0\}$ is regular, by pumping lemma,
there exists an integer~$q$, for any word~$Z\in L_b(u)$ with $|Z|\geq q$, there exists a decomposition~$Z=V_0V_1V_2$~with $|V_1|\geq1,|V_0V_1|\leq q$, such that $V_0V_1^nV_2\in L_b(\mathbf{u})$ for all $n\geq0$. Hence, there exists an increasing integer sequence $\{s_n\}_{n\geq0}$ such that
\begin{equation}\label{2}
  (v_{s_n})_b=([u_0u_1u_2\cdots u_{s_n}]_b)_b=V_0V_1^nV_2,~n\geq0.
\end{equation}

Assume $|V_1|=p$, by claim $(\ref{1})$ and formula $(\ref{2})$, we have if $s_n\geq N$, then $s_{n+1}=s_n+p.$ Hence, for any large integer $n$, there exists $n_0$ and $0\leq i\leq p-1$ such that $n=s_{n_0}+i$. Note that $v_{s_{n_0}+i}=b^i v_{s_{n_0}}+[u_{s_{n_0}+1}\cdots u_{s_{n_0}+i}]_b$ and $(v_{s_{n_0}})_b=V_0V_1^{n_0}V_2$. Thus, each sufficient long word $(v_{s_{n_0}+i})_b$ has the form $W_0W_1^mW_2$, where $m\leq n_0$ is maximal and $|W_0|=|V_0|,|W_1|=|V_1|=p, |W_2|=(n_0-m)p+|V_2|+i.$

Since there are finitely many values of $[u_{s_{n_0}+1}\cdots u_{s_{n_0}+i}]_b$ for $0\leq i<p$, there are finitely many forms of $W_0,W_1,W_2$.  For  simplicity, we assume $p=2$ and $u_{s_{n}+1}\in\{\alpha,\beta\}\subseteq\{0,1,\cdots,B-1\}$. Suppose that $(v_{s_n})_b=V_0V_1^nV_2\subseteq \Sigma_b^*$ for all $n\geq0$, then we assume $(bv_{s_{n}}+\alpha)_b=U_0U_1^{m}U_2\subseteq \Sigma_b^*$ for some $m$ and $(bv_{s_{n}}+\beta)_b=Q_0Q_1^{m}Q_2\subseteq \Sigma_b^*$ for some $m$. Note that $|U_0|=|Q_0|=|V_0|$ and $|U_1|=|Q_1|=|V_1|=2$. Hence, each sufficiently long word of $L_b(\mathbf{u})$ can be written as one of following forms: $V_0V_1^nV_2,U_0U_1^{n}U_2$, or $Q_0Q_1^{n}Q_2$.

If a language differs only in finitely many elements from a regular language, then it is also regular. Hence, without loss of generality, we assume all words of $L_b(\mathbf{u})$ have one of following forms: $V_0V_1^nV_2,U_0U_1^{n}U_2$, or $Q_0Q_1^{n}Q_2$. Define a morphism $\sigma$ from $\{a,c,d,e,f,g,x,y,z\}$ to $\Sigma_b^*$ by $a\mapsto U_0, c\mapsto U_1, d\mapsto U_2, e\mapsto V_0, f\mapsto V_1, g\mapsto V_2, x\mapsto Q_0, y\mapsto Q_1, z\mapsto Q_2$. If $U_0U_1^{m}U_2\in L_b(\mathbf{u})$, then the word~$ac^md$~appears in~$\sigma^{-1}(L_b(\mathbf{u}))$. Note that $\sigma^{-1}(L_b(\mathbf{u}))$ is regular, since $L_b(\mathbf{u})$ is regular. By pumping lemma, if $ac^md\in \sigma^{-1}(L_b(\mathbf{u}))$, then there exists a decomposition~$ac^md=(ac^r)(c^s)(c^{m-r-s}d)$, and $(ac^r)((c^s)^n)(c^{m-r-s}d)\in\sigma^{-1}(L_b(\mathbf{u}))$ for all $n\geq0$, i.e.,
 $U_0U_1^{m+(n-1)s}U_2\in L_b(u)$ for all $n\geq0$. In other words, if $u_{s_{n_0}+1}=\alpha$ for some $n_0$, then $u_{s_{n_0}+2ns+1}=\alpha$ for all $n\geq0$, i.e., $\alpha$ appears periodically with period $2s$. Similarly, $\beta$ also appears periodically. Thus, $\{u_{s_{n}+1}\}_{n\geq0}$ is ultimately periodic, since  $s_{n+1}=s_n+2$. Note from formula $(\ref{2}$) that
 $$[u_{s_n+1}u_{s_{n+1}}]_b=[V_0V_1^{n+1}V_2]_b-b^2[V_0V_1^nV_2]_b=[V_1V_2]_b-b^2[V_2]_b$$
 is constant for $n$ is large enough. So~$\{u_{s_{n}+2}\}_{n\geq0}$ is also ultimately periodic. Hence the sequence~$\mathbf{u}=\{u_n\}_{n\geq0}$ is ultimately periodic.
\end{proof}

In fact, Theorem \ref{regular language} gives a method to prove that a language is non-regular. We give an example to end this section.
\begin{example}
Let $\mathbf{u}$ be a Thue-Morse block sequence over $\{A,B\}$, where $A=10,B=02$, i.e., $\mathbf{u}=u_0u_1u_2\cdots=1002021002101002\cdots$. Clearly, $\mathbf{u}$ is non-ultimately periodic. Hence, by Theorem \ref{regular language}, the language $L_2(\mathbf{u})$ is not regular.
\end{example}

\section{Proof of Theorem \ref{main}}
Let $\alpha,\beta\in\mathbb{R}$ and $b\geq2$ be an integer, assume $\beta<\alpha<b$. For any integer $k\geq1$, define~
\begin{equation}\label{rk}
  r_k:=\left\lfloor b\left\{\frac{b^k-\beta}{\alpha}\right\}+(b-1)\left\{\frac{\beta}{\alpha}\right\}\right\rfloor.
\end{equation}
In this section, we firstly study the non-negative integer sequence $\{r_k\}_{k\geq1}$.

Note that $\left\{\frac{\beta}{\alpha}\right\}=\frac{\beta}{\alpha}$. Assume $\frac{1}{\alpha},\frac{\beta}{\alpha}$ can be represented in the following forms:
$$\frac{1}{\alpha} := \left\lfloor \frac{1}{\alpha}\right\rfloor+\sum_{i\geq1}\frac{\alpha_i}{b^i},\quad \frac{\beta}{\alpha} := \sum_{i\geq1}\frac{\beta_i}{b^i},$$
where $\alpha_i,\beta_i \in \Sigma_b$ for all $i\geq1$.
For any integer~$k\geq0$, define a proposition, denoted by $P_k$,
$$P_k:~0.\alpha_{k+1}\alpha_{k+2}\cdots:=\sum_{i\geq1}\frac{\alpha_{k+i}}{b^i}\geq 0.\beta_{1}\beta_{2}\cdots:=\sum_{i\geq1}\frac{\beta_i}{b^i}.$$
Then its converse proposition, denoted by~$\overline{P_k}$, is
$$\overline{P_k}:~0.\alpha_{k+1}\alpha_{k+2}\cdots:=
\sum_{i\geq1}\frac{\alpha_{k+i}}{b^i}<0.\beta_{1}\beta_{2}\cdots:=\sum_{i\geq1}\frac{\beta_i}{b^i}.$$

\begin{lemma}\label{lemmar1}
For any integer~$k\geq1$, we have
$$
  r_k=\left\{
        \begin{array}{ll}
          \alpha_{k+1}, & \hbox{if~$P_k,P_{k+1}$~hold;} \\
          \alpha_{k+1}-1, & \hbox{if~$P_k,\overline{P_{k+1}}$~hold;} \\
          b+\alpha_{k+1}, & \hbox{if~$\overline{P_k},P_{k+1}$~hold;} \\
          b+\alpha_{k+1}-1, & \hbox{if~$\overline{P_k},\overline{P_{k+1}}$~hold.}
        \end{array}
      \right.
$$
\end{lemma}
\begin{proof}
By the representations of $\frac{1}{\alpha},\frac{\beta}{\alpha}$, it is easy to check that for any~$k\geq1$,
$$
\left\{\frac{b^k-\beta}{\alpha}\right\}=\left\{
                      \begin{array}{ll}
                        0.\alpha_{k+1}\alpha_{k+2}\cdots-0.\beta_1\beta_2\cdots, & \hbox{if~$P_k$~holds,} \\
                        1.\alpha_{k+1}\alpha_{k+2}\cdots-0.\beta_1\beta_2\cdots, & \hbox{if~$\overline{P_k}$~holds,}
                      \end{array}
                    \right.
$$
which implies that
$$b\left\{\frac{b^k-\beta}{\alpha}\right\}+(b-1)\left\{\frac{\beta}{\alpha}\right\}=\left\{
                                             \begin{array}{ll}
                                               \alpha_{k+1}.\alpha_{k+2}\cdots-0.\beta_1\beta_2\cdots, & \hbox{if~$P_k$~holds,} \\
                                               b+\alpha_{k+1}.\alpha_{k+2}\cdots-0.\beta_1\beta_2\cdots, & \hbox{if~$\overline{P_k}$~holds.}
                                             \end{array}
                                           \right.
$$

Hence, $r_k=\alpha_{k+1}$ if and only if $P_k,P_{k+1}$~both hold. Similarly, the other cases can be easily obtained.
\end{proof}
\begin{remark}
$r_k\in \Sigma_{2b-1}$ for all integers~$k\geq1$. Moreover, if~$r_k=\alpha_{k+1}-1$, then~$\alpha_{k+1}\geq1$. Hence, $(\alpha_{k+1}-1)_b=\alpha_{k+1}-1.$
\end{remark}
\begin{proof}
If $r_k=\alpha_{k+1}-1,$ then both $P_k$ and $\overline{P_{k+1}}$ hold. Hence, $\sum_{i\geq1}\frac{\alpha_{k+i}}{b^i}\geq\sum_{i\geq1}\frac{\beta_i}{b^i}>\sum_{i\geq1}\frac{\alpha_{k+i+1}}{b^i}.$ So, we have $\sum_{i\geq1}\frac{\alpha_{k+i}-\alpha_{k+i+1}}{b^i}>0$, which implies that $\alpha_{k+1}\geq1.$
\end{proof}
\begin{lemma}\label{lemmar2}
Let $k\geq1$ be an integer, then we have
\begin{enumerate}
  \item $r_k\in\{\alpha_{k+1},b+\alpha_{k+1}\}$ if and only if~$r_{k+1}\in\{\alpha_{k+2},\alpha_{k+2}-1\}.$
  \item $r_k\in\{\alpha_{k+1}-1,b+\alpha_{k+1}-1\}$ if and only if~$r_{k+1}\in\{b+\alpha_{k+2},b+\alpha_{k+2}-1\}$.
\end{enumerate}
\end{lemma}
\begin{proof}
By Lemma \ref{lemmar1}, it is clear that for any integer $k\geq0$, $r_k\in\{\alpha_{k+1},\alpha_{k+1}-1\}$ if and only if $P_{k}$~holds, $r_k\in\{\alpha_{k+1},b+\alpha_{k+1}\}$ if and only if $P_{k+1}$~holds, $r_k\in\{\alpha_{k+1}-1,b+\alpha_{k+1}-1\}$ if and only if $\overline{P_{k+1}}$~holds, $r_k\in\{b+\alpha_{k+1},b+\alpha_{k+1}-1\}$ if and only if $\overline{P_{k}}$~holds.
\end{proof}
\begin{lemma}\label{lemmar3}
Let $k\geq1$ be an integer,  then we have
\begin{enumerate}
  \item If~$r_k\in\{\alpha_{k+1},b+\alpha_{k+1}\}$, then
  $$([r_1r_2\cdots r_k]_b)_b \in \{\alpha_2\cdots\alpha_{k}\alpha_{k+1}, 1\alpha_2\cdots\alpha_{k}\alpha_{k+1}\}.$$
  \item If ~$r_k=\alpha_{k+1}-1$, then
  $$([r_1r_2\cdots r_k]_b)_b \in \{\alpha_2\cdots\alpha_{k}(\alpha_{k+1}-1), 1\alpha_2\cdots\alpha_{k}(\alpha_{k+1}-1)\}.$$
  \item If~$r_k=b+\alpha_{k+1}-1$ with $\alpha_{k+1}\geq1$, then
  $$([r_1r_2\cdots r_k]_b)_b\in\{\alpha_2\cdots\alpha_{k}(\alpha_{k+1}-1),1\alpha_2\cdots\alpha_{k}(\alpha_{k+1}-1)\}.$$
  \item If~$r_k=b+\alpha_{k+1}-1$ with $\alpha_{k+1}=0$, then
  $$([r_1r_2\cdots r_k]_b+1)_b\in\{\alpha_2\cdots\alpha_{k}\alpha_{k+1},1\alpha_2\cdots\alpha_{k}\alpha_{k+1}\}.$$
\end{enumerate}
\end{lemma}
\begin{proof}
We prove this lemma by induction on~$k$. By Lemma \ref{lemmar1}, it is easy to check that the three assertions are true for~$k=1,2$. Now assume it is true for any~$k\leq m$, then we consider the case~$k=m+1$.
\begin{itemize}
  \item If $r_{m+1}\in\{\alpha_{m+2},\alpha_{m+2}-1\}$, then by Lemma \ref{lemmar2},~$r_m\in\{\alpha_{m+1},b+\alpha_{m+1}\}$.
        \begin{itemize}
          \item If $r_{m+1}=\alpha_{m+2}$, then  $(r_{m+1})_b=\alpha_{m+2}$ and $([r_1\cdots r_{m+1}]_b)_b=([r_1\cdots r_{m}]_b)_b(r_{m+1})_b\in\{\alpha_2\cdots\alpha_{m+2},1\alpha_2\cdots\alpha_{m+2}\}$ by induction.
          \item If $r_{m+1}=\alpha_{m+2}-1,$ then  $(r_{m+1})_b=\alpha_{m+2}-1$ and $([r_1\cdots r_{m+1}]_b)_b=([r_1\cdots r_{m}]_b)_b(r_{m+1})_b\in\{\alpha_2\cdots(\alpha_{m+2}-1),1\alpha_2\cdots(\alpha_{m+2}-1)\}$ by induction.
        \end{itemize}
  \item If $r_{m+1}\in\{b+\alpha_{m+2},b+\alpha_{m+2}-1\}$, then by Lemma \ref{lemmar2}, we have ~$r_m\in\{\alpha_{m+1}-1,b+\alpha_{m+1}-1\}$. By induction, we note that
      \begin{equation}\label{f}
        ([r_1r_2\cdots r_m]_b+1)_b\in\{\alpha_2\cdots\alpha_{m+1},1\alpha_2\cdots\alpha_{m+1}\}.
      \end{equation}
        \begin{itemize}
          \item If $r_{m+1}=b+\alpha_{m+2}$, then $(r_{m+1})_b=1\alpha_{m+2}$. Hence, by formula (\ref{f}), we have $([r_1\cdots r_{m+1}]_b)_b=([r_1\cdots r_{m}]_b+1)_b\alpha_{m+2}\in\{\alpha_2\cdots\alpha_{m+2},\\1\alpha_2\cdots\alpha_{m+2}\}$.
          \item If  $r_{m+1}=b+\alpha_{m+2}-1$ with $\alpha_{m+2}\geq1$, then $(r_{m+1})_b=1(\alpha_{m+2}-1)$. Hence, by formula (\ref{f}), we have $([r_1\cdots r_{m+1}]_b)_b=([r_1\cdots r_{m}]_b+1)_b(\alpha_{m+2}-1)\in\{\alpha_2\cdots(\alpha_{m+2}-1),1\alpha_2\cdots(\alpha_{m+2}-1)\}$.
          \item If  $r_{m+1}=b+\alpha_{m+2}-1$ with $\alpha_{m+2}=0$, then $(r_{m+1}+1)_b=1\alpha_{m+2}$. Hence, by formula (\ref{f}), we have $([r_1\cdots r_{m+1}]_b+1)_b=([r_1\cdots r_{m}]_b+1)_b\alpha_{m+2}\in\{\alpha_2\cdots\alpha_{m+2},1\alpha_2\cdots\alpha_{m+2}\}$.
        \end{itemize}
\end{itemize}
Thus, the assertions are true for $k=m+1$, which completes this proof.
\end{proof}
\begin{lemma}\label{periodic}
The sequence $\{r_k\}_{k\geq1}$ is ultimately periodic if and only if $\alpha$ is rational.
\end{lemma}
\begin{proof}
If~$r_k=b+\alpha_{k+1}-1$ with $\alpha_{k+1}=0$ for all $k\geq N$ for some $N$, then $\alpha_{k}$ is ultimately periodic and  $\alpha$ is rational.  If there exist infinitely many $k$ such that $r_k\neq b+\alpha_{k+1}-1$ with $\alpha_{k+1}=0$, then, by Lemma \ref{lemmar3},
\begin{eqnarray*}
  ([r_1\cdots r_k]_b)_b &\in& \{\alpha_2\cdots\alpha_{k}\alpha_{k+1},\alpha_2\cdots\alpha_{k}(\alpha_{k+1}-1), \\
   && 1\alpha_2\cdots\alpha_{k}\alpha_{k+1},1\alpha_2\cdots\alpha_{k}(\alpha_{k+1}-1)\}.
\end{eqnarray*}
Hence, if the sequence $\{r_k\}_{k\geq1}$~is ultimately periodic, then the sequence $\{\alpha_k\}_{k\geq2}$ is also ultimately periodic, which implies that $\alpha$ is rational.

Conversely, if $\alpha=\frac{p}{q}$, where $p,q$ are integers and $(p,q)=1$, then, $$r_k=\left\lfloor b\left\{\frac{qb^k-q\beta}{p}\right\}+(b-1)\frac{q\beta}{p}\right\rfloor.$$
If $b^i\equiv b^j (\bmod~p)$, then $\left\{\frac{qb^i-q\beta}{p}\right\}=\left\{\frac{qb^j-q\beta}{p}\right\}$, i.e., $r_i=r_j$. Hence, the sequence $\{r_k\}_{k\geq1}$ is ultimately periodic.
\end{proof}

Using Lemma \ref{periodic} and Theorem \ref{regular language}, we are going to prove Theorem \ref{main}.
\begin{proof}[Proof of Theorem \ref{main}]
Let~$u_n=\lfloor\log_b(n\alpha+\beta)\rfloor$. Without loss of generality, we assume $\beta<\alpha<b$.  In fact,
if $\beta\geq \alpha$, i.e., $\beta=m\alpha+r$ for some integer $m\geq1$ and $0\leq r<\alpha$, then $\lfloor\log_b(n\alpha+\beta)\rfloor=\lfloor\log_b((n+m)\alpha+r)\rfloor$.
If $\alpha\geq b$, i.e., $b^m\leq \alpha<b^{m+1}$ for some integer $m\geq1$, then $\frac{\alpha}{b^m}<b$ and $\lfloor\log_b(n\alpha+\beta)\rfloor=\lfloor m+\log_b(n\frac{\alpha}{b^m}+\frac{\beta}{b^m})\rfloor=m+\lfloor \log_b(n\frac{\alpha}{b^m}+\frac{\beta}{b^m})\rfloor$.

Set $v_n:=u_{n+1}-u_n$, then
$$v_n\leq\log_b((n+1)\alpha+\beta)-\log_b(n\alpha +\beta)+1=\log_b\frac{(n+1)\alpha+\beta}{n\alpha +\beta}+1.$$
When $n$ tends to infinite, $\log_b\frac{(n+1)\alpha+\beta}{n\alpha +\beta}$ tends to $0$. Thus, there exits an integer~$N_0>1$ such that $v_n\leq 1$~for all~$n>N_0$, i.e.,~$v_n\in\{0,1\}$~for all~$n>N_0$.
It is easy to check that ~$v_n=1$ if and only if there exists an unique integer~$k$~satisfying
$k-1\leq\log_b(n\alpha+\beta)<k\leq\log_b((n+1)\alpha+\beta)<k+1$
if and only if $n\in[\frac{b^k-\beta}{\alpha}-1,\frac{b^k-\beta}{\alpha})$.
If there exist two integers $k_1,k_2$ such that $\frac{b^{k_1}-\beta}{\alpha}=m_1$ and $\frac{b^{k_2}-\beta}{\alpha}=m_2$ for some integers $m_1,m_2$, then $\alpha=\frac{b^{k_2}-b^{k_1}}{m_2-m_1}\in \mathbb{Q}$. Hence, we always assume $\frac{b^k-\beta}{\alpha}$ are not integers for all $k$ and ~$v_n=1$ if and only if $n=\lfloor\frac{b^k-\beta}{\alpha}\rfloor$.

If a sequence differs only in finitely many elements from a regular (resp. automatic) sequence, then it is also regular (resp. automatic).
Hence, without loss of generality, we assume
$$
v_n=\left\{
      \begin{array}{ll}
        1, & \hbox{if~$\exists~k\geq1$ such that $n=\lfloor\frac{b^k-\beta}{\alpha}\rfloor$;} \\
        0, & \hbox{othewise.}
      \end{array}
    \right.
$$

For any integer $k\geq1$, let~$c_k:=\lfloor\frac{b^k-\beta}{\alpha}\rfloor$. It is easy to check that $c_{k+1}=bc_k+r_{k}$, where
~$r_k$ is defined in formula (\ref{rk}). Note that $c_k\geq c_1\geq0$ for all $k\geq1$. Hence, for any integer $b,k\geq2$, we have $c_k=[c_1r_1\cdots r_{k-1}]_b.$

Define a language $L:=\{(c_k)_b:k\geq1\}$, then
\begin{equation}\label{regular}
  L=\left\{\left([c_1r_1r_2\cdots r_k]_b\right)_b:k\geq1\right\}\cup\left\{(c_1)_b\right\}.
\end{equation}

Note in \cite{Jp} that the running-sum of $k$-regular sequences is also $k$-regular and  a $k$-regular sequence  taking finitely many values if and only if it is $k$-automatic. Hence, by Theorem \ref{regular language}, Lemma \ref{periodic} and formula (\ref{regular}), we have
\begin{eqnarray*}
  \{u_n\}_{n\geq0}~ is~regular &\Leftrightarrow& \{v_n\}_{n\geq0}~is~automatic \\
   &\Leftrightarrow& \{c_k:k\geq1\}~is~recognisable \\
   &\Leftrightarrow& L=\{(c_k)_b:k\geq1\}~is ~regular \\
   &\Leftrightarrow& \left\{\left([c_1r_1r_2\cdots r_k]_b\right)_b:k\geq1\right\}~ is~regular\\
   &\Leftrightarrow& \{r_k\}_{k\geq1}~ is~ ultimately~ periodic\\
   &\Leftrightarrow& \alpha\in\mathbb{Q},
\end{eqnarray*}
which completes our proof.
\end{proof}

For any integer $k\geq0$, define $f_k=\sharp\{n:\lfloor\log_{b}(\alpha n+\beta)\rfloor=k\}$, where~$\sharp$~denotes the cardinality of a set.
Then we have the following description of $f_k$.
 \begin{theorem}
The sequence~$\{f_{k+1}-bf_k\}_{k\geq0}$~is ultimately periodic if and only if~$\alpha\in\mathbb{Q}$.
\end{theorem}
\begin{proof}
Let $u_n=\lfloor\log_{b}(\alpha n+\beta)\rfloor$ and $v_n=u_{n+1}-u_n$. By Theorem \ref{main}, there exists an integer~$n_0$~such that $v_n\in\{0,1\}$ for all~$n\geq n_0$.
Hence, there exists an integer~$m_0$ such that for any integer~$k>u_{n_0}$,
$f_k=c_{k+m_0+1}-c_{k+m_0}.$
Note that for any~$k\geq1$, $c_{k+1}-bc_k=r_k$, where $r_k$ is defined by formula \ref{rk}.
Thus, for any integer~$k>u_{n_0}$, we have
\begin{eqnarray}\label{d_n}
\nonumber  f_{k+1}-bf_k &=& c_{k+m_0+2}-c_{k+m_0+1}-b(c_{k+m_0+1}-c_{k+m_0})  \\
             &=&  r_{k+m_0+1}-r_{k+m_0}.
\end{eqnarray}

If~$\alpha\in\mathbb{Q}$, then, by Lemma \ref{periodic}, $\{r_k\}_{k\geq1}$ is ultimately periodic. Hence, by formula (\ref{d_n}),  The sequence~$\{f_{k+1}-bf_k\}_{k\geq0}$~is ultimately periodic.

On the other hand, let $d_k=f_{k+1}-bf_k$, assume $\{d_k\}_{k\geq0}$ is ultimately periodic, i.e., $\exists~N, p$ such that $d_{i+p}=d_i$ for all $i\geq N$. Since finite terms do not change the periodic property of sequence, by formula (\ref{d_n}), we assume $d_k=r_{k+m_0+1}-r_{k+m_0}$ for all $k\geq 0$. Hence, for any integer $k\geq0$, $r_{k+m_0+1}=r_{m_0}+\sum_{i\geq0}^{k}d_i$.

Now, consider the sequence $\{\sum_{i\geq0}^{k}d_i\}_{k\geq0}$. Since for any~$k\geq1$, $0\leq r_k\leq 2b-1$,  we assume $\sum_{i\geq0}^{k}d_i=r_{k+m_0+1}-r_m$  takes at most $M$ values. Then, there must be two  elements of $\sum_{i=0}^{N+jp-1}d_i(0\leq j\leq M)$ are the same. In other words, there exist two integers $0\leq s<t\leq M$ such that $\sum_{i=0}^{N+sp-1}d_i=\sum_{i=0}^{N+tp-1}d_i$. Hence, the sequence  $\{\sum_{i\geq0}^{k}d_i\}_{k\geq0}$ is ultimately periodic, which implies that $\{r_{k}\}_{k\geq1}$ is ultimately periodic. Thus, $\alpha\in\mathbb{Q}$.
\end{proof}

\medskip

\noindent\textbf{Acknowledgements. }
The authors wish to thank  Professor Zhi-Ying Wen for inviting them to visit
the Morningside Center of Mathematics, Chinese Academy of Sciences.
They also wish to thank Professor Qing-Hui Liu for his helpful suggestions

\end{document}